\documentclass[a4paper]{amsart}
\usepackage{amsmath}
\usepackage{amssymb}
\usepackage{amscd}
\usepackage{epsfig}
\usepackage{amsxtra}
\usepackage{pifont}
\usepackage{mathabx}
\usepackage{MnSymbol}
\usepackage{accents}
\usepackage{xcolor}
\usepackage{scalerel,stackengine}
\stackMath
\newcommand\reallywidehat[1]{%
\savestack{\tmpbox}{\stretchto{%
  \scaleto{%
    \scalerel*[\widthof{\ensuremath{#1}}]{\kern-.6pt\bigwedge\kern-.6pt}%
    {\rule[-\textheight/2]{1ex}{\textheight}}
  }{\textheight}%
}{0.5ex}}%
\stackon[1pt]{#1}{\tmpbox}%
}

\numberwithin{equation}{section}

\newcommand{\supp}{\mbox{\rm supp}}
\newcommand{\dens}{\mbox{\rm dens}}

\newcommand{\R}{{\mathbb R}}

\newcommand{\C}{{\mathbb C}}

\newcommand{\mig}{{\mathcal M}^\infty}

\newcommand{\cL}{{\mathcal L}}

\newcommand{\hG}{\widehat{G}}

\newcommand{\oplam}{\mbox{\Large $\curlywedge$}}

 \newtheorem{theorem}{Theorem}[section]
 \newtheorem{lemma}[theorem]{Lemma}
 \newtheorem{prop}[theorem]{Proposition}
 \newtheorem{cor}[theorem]{Corollary}

 \newtheorem{defi}[theorem]{Definition}

\newtheoremstyle{rremark}%
       {1.8ex\@plus1ex}                
       {2.1ex\@plus1ex\@minus.5ex}      
       {\normalfont}        
       {0pt}                   
       {\bfseries}           
       {.}                  
       {.5em}               
       {}                   

\theoremstyle{rremark}
\newtheorem{remark}[theorem]{Remark}

\begin{document}
\title{Pure point diffraction and Poisson summation }

\author{Christoph Richard}
\address{Department f\"{u}r Mathematik, Friedrich-Alexander-Universit\"{a}t Erlangen-N\"{u}rnberg,
Cauerstrasse 11, 91058 Erlangen, Germany}
\email{christoph.richard@fau.de}

\author{Nicolae Strungaru}
\address{Department of Mathematical Sciences, MacEwan University \\
10700 “ 104 Avenue, Edmonton, AB, T5J 4S2\\
and \\
Department of Mathematics\\
Trent University \\
Peterborough, ON
and \\
Institute of Mathematics ``Simon Stoilow''\\
Bucharest, Romania}
\email{strungarun@macewan.ca}
\urladdr{http://academic.macewan.ca/strungarun/}

 \maketitle

\begin{abstract}
We prove that the diffraction formula for regular model sets is equivalent to the Poisson Summation Formula for the underlying lattice. This is achieved using Fourier analysis of unbounded measures on locally compact abelian groups, as developed by Argabright and de Lamadrid. We also discuss related diffraction results for certain classes of non-regular so-called weak model sets.
\end{abstract}

\renewcommand{\thefootnote}{\fnsymbol{footnote}}
\renewcommand{\thefootnote}{\arabic{footnote}}

\section{Introduction}

Research triggered by the experimental discovery of quasicrystals, see~\cite{BG2} for a recent mathematical monograph, provided examples of non-periodic structures with long-range order. Their diffraction spectrum consists of Bragg peaks only. A mathematical abstraction of these examples are so-called regular model sets, which are certain projections of a subset of a higher-dimensional lattice. In fact, model sets have been introduced and intensively studied before the discovery of quasicrystals by Meyer \cite{me70,me72}, and they have later been re-investigated and advocated by Moody, see e.g.~\cite{RVM3}.

A central result in mathematical diffraction theory states that regular model sets have a pure point diffraction spectrum. The particular form of their Bragg peaks has been proved in the Euclidean setting by Hof \cite{Hof1,Hof3}, using the underlying lattice Poisson Summation Formula (PSF). When it became clear that examples with a non-Euclidean embedding space exist, such as limit-periodic model sets \cite{BMS}, their diffraction has been studied in general $\sigma$-compact locally compact abelian (LCA) groups. Only then pure point diffractivity had been fully established. However the corresponding proofs did not rely on the PSF, but instead used dynamical systems, see e.g.~Schlottmann \cite{Martin2} for so-called repetitive regular model sets, and \cite{Nicu14, KR} for recent results in that direction.
On the other hand, it had been remarked by Lagarias that an alternative proof based on the PSF should be possible. In the Euclidean setting, such a proof has recently been given by Baake--Grimm \cite[Thm.~9.4]{BG2} for a subclass of regular model sets. Yet another approach to pure point diffractivity uses almost periodic measures, see e.g.~Solomyak \cite{So98}, Baake--Moody \cite{BM} and the recent work \cite{NS11}. These results have also been extended to weighted Dirac combs with weight functions on the embedding space of sufficiently fast decay, which are called admissible in \cite{CR,LR}.

In this article, we will reprove the diffraction formula for regular model sets on general $\sigma$-compact LCA groups using the PSF of the underlying lattice.
Working beyond Euclidean space, we cannot follow the tempered distribution approach as in \cite{Hof1, BG2}, but rely instead on Fourier theory of unbounded measures, as in Argabright and de Lamadrid~\cite{ARMA1, BF}. Whereas working with measures only seems a severe restriction, this is in fact natural, as diffraction may be described by a measure. Indeed, the diffraction measure of the infinite idealisation of a finite specimen decomposes into a pure point and a continuous part, which correspond to the Bragg peaks and to the continuous component in the diffraction picture.

We will extend elements of Hof's proof to the non-Euclidean setting. A crucial ingredient in Hof's approach (and in all later approaches to pure point diffractivity) is a certain uniform distribution result for regular model sets.  Proofs of uniform distribution in the general setting are based on geometric \cite{sch98} or dynamical systems arguments \cite{RVM1}. We will show that uniform distribution actually follows from the underlying lattice PSF. In the Euclidean setting, this has already been argued by Meyer \cite{me70}, see also \cite[Sec.~V.7.3]{me72}, \cite[Prop.~5.1]{mm10}.
This insight allows us to show that the lattice PSF and the diffraction formula for regular model sets can be derived from one another.

Thus our proofs shed also some new light on the Euclidean case. They justify an opinion sometimes expressed by experimentalists, that diffraction properties of quasicrystals should be deducible from those of the underlying lattice. Indeed, the uniform distribution result can be seen as a consequence of the lattice PSF. In fact, the diffraction formula of the underlying lattice is equivalent to the diffraction formula for regular model sets in this setting.

Our approach also indicates that problems equivalent to the PSF, such as function reconstruction via sampling \cite{But}, may also be successfully analysed for model sets in the abstract setting. It further suggests a method for extending cut-and-project schemes beyond the abelian case, which might provide an alternative to the recent approach in \cite{BHP16}. We remark that, for our main results, the essential property of the underlying lattice is Fourier transformability, not the lattice property itself. Indeed our approach can be extended to a large class of underlying Fourier transformable measures \cite{RicStr}.

Let us describe the structure of this article. After reviewing mathematical diffraction theory, we recall the Fourier analysis of unbounded measures in Section~\ref{prel}. Here we are particularly interested in generalised PSF for measures, which extend the lattice case, and in double transformability. In Section~\ref{sec:fapm} we consider weighted model sets and discuss their Fourier analysis. We prove that, for certain weight functions, their generalized PSF is equivalent to the PSF of the underlying lattice, and we show uniform distribution using the lattice PSF. In Section~\ref{sec:pcp}, we will use the so-called autocorrelation measure to study the diffraction of regular model sets. We show that the diffraction formula for regular model sets can be obtained by combining the generalised PSF and the density formula. Thus the diffraction formula can be obtained by a double application of the PSF for the underlying lattice. We also show that the lattice PSF follows from the diffraction formula for regular model sets. We close with some remarks on pure point diffraction for non-regular model sets and Meyer sets, and by proving double transformability of translation bounded transformable measures with Meyer set support.

\section{Elements of diffraction theory}\label{sec:ele}

We give a focussed and refined introduction to diffraction theory, see \cite{G63, C} for physical background and \cite[Section~9.1.2]{BG2} for mathematical background in the Euclidean case.
Diffraction of X-rays by matter results from scattering by the individual atoms and interference between the scattering waves.  For a finite sample with atom position set $\Lambda\subset \mathbb R^3$, one defines the structure factor $F(\boldsymbol s)=\sum_{\boldsymbol p\in\Lambda} f_{\boldsymbol p}e^{-2\pi \dot \imath\, \boldsymbol s\cdot \boldsymbol p}$, where the so-called scattering factor $f_{\boldsymbol p}$ is a complex weight associated with an atom at $\boldsymbol p$. Assume that the scattering is elastically, that the incident beam is a plane wave with wave vector $\boldsymbol k_0$, and assume that one measures diffraction at a distance $\boldsymbol r$ very large in comparison to the sample size. Then the observed intensity $I(\boldsymbol r)$ of diffraction at distance $\boldsymbol r$ is approximately given by
\begin{displaymath}
I(\boldsymbol r) = \frac{A}{|\boldsymbol r|^2} |F(\boldsymbol k-\boldsymbol k_0)|^2,
\end{displaymath}
where  $\boldsymbol k=|\boldsymbol k_0|\cdot\boldsymbol r/|\boldsymbol r|$, and $A$ is some positive normalisation constant.

For a measure-theoretic description, assume that $\Lambda\subset \mathbb R^d$ is uniformly discrete and
consider the Dirac comb $\omega=\delta_\Lambda=\sum_{p\in\Lambda} \delta_p$ of  $\Lambda$. We  assume for simplicity that all atoms have equal scattering factors $f_p=1$.
We want to infer the diffraction of $\omega$ from finite samples $\omega_n=\omega|_{B_n}$, where we restrict to centered balls $B_n$ of radius $n$.
Diffraction of the finite measure $\omega_n$ can be rephrased using a so-called Wiener diagram
\begin{displaymath}
\begin{CD}
\omega_n @>*>> \omega_n*\widetilde{\omega_n}\\
@V{\mathcal F}VV @VV{\mathcal F}V\\
{\widehat{\omega_n}} @>{|\cdot|^2}>> \widehat{\omega_n}\cdot \overline{\widehat{\omega_n}}
\end{CD}
\end{displaymath}
Here $\mathcal F$ denotes the Fourier transform, $*$ denotes convolution of measures, and  the reflected measure $\widetilde{\omega}$ is defined via $\widetilde{\omega} (f)=\overline{\omega(\widetilde f)}$ with $\widetilde f(x)=\overline{f(-x)}$. Due to the convolution theorem, the diffraction of $\omega_n$ may thus alternatively be computed as the Fourier transform of the positive definite so-called autocorrelation measure $\gamma_n=\omega_n*\widetilde{\omega_n}$. 

For $\omega$ unbounded instead of $\omega_n$, the Wiener diagram has no direct measure-theoretic interpretation. Whereas convolution is defined only if one measure is bounded, an autocorrelation $\gamma$ of $\omega$ may however be interpreted
via the so-called Eberlein convolution
%
\begin{displaymath}
\gamma=\omega \circledast \widetilde \omega:=\lim_{n\to\infty} \frac{1}{\mathrm{vol}(B_n)} \omega_n*\widetilde{\omega_n},
\end{displaymath}
if this vague limit exists.
Moreover, a Fourier transform of $\omega$ may make sense via tempered distributions but not as a measure if $\omega$ is non-periodic. Indeed, if both $\omega$ and $\widehat\omega$ are measures with uniformly discrete support, and if $\widehat\omega$ is positive, then $\omega$ must be supported within a finite union of lattice translates \cite{LO, Fa15}.
No difficulty arises for $\omega=\delta_\Lambda$ a lattice Dirac comb. Its Fourier transform is the measure $\widehat \omega=\mathrm{dens}(\Lambda)\cdot\delta_{\Lambda_0}$, where $\Lambda_0$ is the lattice dual to $\Lambda$. Note that this measure equation is a version of the classical PSF.  As one also has $\gamma=\mathrm{dens}(\Lambda)\cdot\omega$,  the modified Wiener diagram commutes in that situation.

More generally, if $\omega$ is not Fourier transformable as a measure, it is still  possible to analyse diffraction from a measure theoretic viewpoint if an autocorrelation $\gamma$ of $\omega$ exists. This will be true in all examples below.
In that situation, $\gamma$ will be transformable due to positive definiteness, and its Fourier transform $\widehat \gamma$ will be a positive measure. One can then infer diffraction properties of $\Lambda$ from the Lebesgue decomposition of $\widehat\gamma$. The pure point part describes the Bragg peaks, and the continuous part describes the diffuse background in the diffraction picture. Moreover, diffraction of $\Lambda$ can indeed be inferred from finite samples as the Fourier transform, when restricted to positive definite measures, is continuous, see \cite[Thm.~4.16]{BF} and \cite[Lemma 1.26]{MoSt}. From a physical viewpoint, this justifies why one may approximate a large finite sample by its infinite idealisation.

Our main result Theorem~\ref{theo:main} states that Dirac combs of regular model sets, which includes lattice Dirac combs, are pure point diffractive. In that case, their diffraction amplitudes are computed as in the finite measure case by ``squaring the Fourier-Bohr coefficients''. Moreover a modified Wiener diagram holds for a large class of weighted model sets including lattice Dirac combs, see Remark~\ref{rem:mwd}.

\section{Fourier transformability}\label{prel}

Let us fix our notation. $G$ stands for an arbitrary LCA group. For subsection \ref{FB coefficients} and all the results dependent on this subsection we also will need the assumption that $G$ is $\sigma$-compact, but for most results this is not necessary. In general, $G$ will not be assumed to be $\sigma$-compact,  whenever when we need this extra assumption we will clearly state this.
A Haar measure on $G$ will be denoted by $\theta_G$. We denote by $C_c(G)$ the space of continuous, compactly supported functions on $G$, and by $C_U(G)$ the space of uniformly continuous and bounded functions on $G$.  For $f\in L^1(G)$ we denote its Fourier transform by $\widehat f$ and its inverse Fourier transform by $\widecheck f$. Let $\hG$ denote the Pontryagin dual of $G$. Given any LCA group $G$ with Haar measure $\theta_G$, we always choose the Plancherel measure $\theta_{\widehat G}$ on the dual group $\widehat G$, i.e., the Haar measure such that the Plancherel theorem \cite[Thm.~3.4.8]{DE} holds.
Let $\mathcal M(G)$ denote the set of complex regular Radon measures on $G$.
 A measure $\mu\in\mathcal M(G)$ is \textit{translation bounded} if $\sup\{|\mu|(t+K)\,|\, t\in G\}<\infty$ for every compact $K\subset G$, where $|\mu|\in\mathcal M(G)$ is the variation measure of $\mu$. Let $\mig(G)\subset \mathcal M(G)$ denote the set of translation bounded complex regular Radon measures on $G$. If for some $1\le p\le\infty$ a function $f:G\to\mathbb C$ satisfies  $f\in L^p(|\mu|)$, we write $f\in L^p(\mu)$ for convenience.

\subsection{Poisson Summation Formula for a lattice}\label{sec:lat}

Let $L$ be  a lattice in $G$, i.e., a discrete, co-compact subgroup of $G$, and consider the lattice Dirac comb $\mu=\delta_L$.
The Dirac comb $\delta_L$ is a translation bounded  positive measure. Normalise the Haar measure on $G/L$ such that the Weil formula for the disintegration over normal subgroups \cite[Eqn.~(3.3.10)]{Rei2} holds.
Let us denote by $L_0\subset \widehat G$ the annihilator of $L$ in $\widehat G$, i.e.,
\begin{displaymath}
L_0 = \{ \chi \in \widehat{G} : \chi(x)= 1 \text{ for all } x \in L \}.
\end{displaymath}
By Pontryagin duality, the annihilator $L_0\cong \widehat{G/L}$ of $L$ is a lattice in $\widehat{G}$. It is called the lattice dual to $L$. Its Dirac comb $\delta_{L_0}$ is a translation bounded positive measure.
As Haar measure on $L$ we choose the counting measure $\theta_L=\delta_L$, and on $L_0$ we choose the Plancherel measure $\theta_{L_0}$ with respect to $G/L$. It is given by $\theta_{L_0}=\mathrm{dens}(L)\cdot\delta_{L_0}$.
Consider $KL(G):= \{ f \in C_c(G) : \widehat{f} \in L^1(\widehat{G}) \}$. The following result is well known, see e.g.~ \cite[Thm.~5.5.2]{Rei2} and \cite[Thm.~3.6.3]{DE}.

\begin{theorem}[lattice PSF]\label{PSFL-I} Let $L\subset G$ be a lattice in $G$, and let $L_0\subset \widehat G$ be its dual lattice. Then the Poisson Summation Formula
\begin{displaymath}
\langle \delta_L, f\rangle \,=\, \mathrm{dens}(L)\cdot\langle \delta_{L_0}, \widecheck{f} \rangle
\end{displaymath}
holds for all $f \in KL(G)$. \qed
\end{theorem}

\begin{remark}
 In fact the lattice PSF also holds for sufficiently decaying $f$ of unbounded support, but we will not consider such functions in this manuscript. Compare however \cite[Sec.~9.2]{BG2} for a corresponding result in the Euclidean setting.
One may regard the measure $\mathrm{dens}(L)\cdot \delta_{L_0}$ as the Fourier transform of the measure $\delta_L$. Note that as a consequence of the lattice PSF, $f\in KL(G)$ has a lattice integrable Fourier transform, i.e., $\widehat f|_{L_0}\in L^1(L_0)$. These two observations motivate the definition of the Fourier transform of a measure in the following section, see also Proposition~\ref{charFT} (ii).
\end{remark}

\subsection{Fourier transforms as measures}

As usual, compare e.g. \cite[Sec.~1.1.6]{RUD}, convolution is for $f,g\in L^1(G)$ defined by  $f*g(x)=\int_G f(y)g(x-y) {\rm d}\theta_G(y)$. We have $f*g=g*f$. We also use $\widetilde{f}(x)=\overline{f(-x)}$, which defines a unitary representation of $G$ on the Hilbert space $L^2(G, \theta_G)$, see \cite{HN98} for background.
Hence $f\in L^2(G, \theta_G)$ satisfies $f*\widetilde f\in P(G)$, where $P(G)$ denotes the set of continuous positive definite functions on $G$.
This implies that $f*\widetilde f$ is Fourier transformable by Bochner's theorem \cite[Thm.~4.4.19]{Rei2}, and we have $\widehat{f*\widetilde f}=|\widehat f|^2$.
We recall the definition of transformability of a measure \cite[Sec.~2]{ARMA1}.

\begin{defi}[Fourier transform] A measure $\mu\in \mathcal M(G)$ is \emph{transformable} if there exists a measure $\widehat{\mu}\in\mathcal M(\hG)$ such that for all $f \in C_c(G)$ we have $\widecheck{f} \in L^2(\widehat{\mu})$ and
\begin{displaymath}
\langle \mu, f*\widetilde{f}\rangle \, =\, \langle\widehat{\mu}, |\widecheck{f}|^2 \rangle \,.
\end{displaymath}
In this case, $\widehat{\mu}$ is called the \emph{Fourier transform of $\mu$}.

\end{defi}

\begin{remark}\label{rem:FTprop}
In the Euclidean setting, the Fourier transform is often considered as an appropriate tempered distribution. Here the Fourier transform $\widehat \mu$ is even required to be a measure. Such $\widehat\mu$ is uniquely determined if it exists \cite[Thm.~2.1]{ARMA1}. Moreover $\widehat\mu$ is then translation bounded \cite[Thm.~2.5]{ARMA1}.
\end{remark}

\begin{remark}\label{rem:ftex}
The above definition generalises the Fourier transform of functions. Indeed, examples of transformable measures are given by $\mu=f\cdot\theta_G$, where $f\in P(G)$ or $f\in L^p(G)$ for $1\le p\le 2$, see \cite[Thm.~2.2]{ARMA1}.
Every finite measure is transformable.
A measure $\mu\in \mathcal M(G)$ is called positive definite if $\int_G f*\widetilde f (x) \, {\rm d}\mu(x)\ge0$ for every $f\in C_c(G)$.  As a consequence of Bochner's theorem,
every positive definite measure is transformable \cite[Thm.~4.1]{ARMA1}.
Any Haar measure on a closed subgroup of $G$ is, as a measure on $G$, positive definite and hence transformable \cite[Prop.~6.2, Cor.~6.2]{ARMA1}. For the Haar measure on a closed subgroup of $G$, the definition of Fourier transform reduces to a version of the classical PSF, compare Section~\ref{sec:lat} and Proposition~\ref{charFT}.  Thus, transformability expresses that the measure satisfies some generalised PSF.
\end{remark}

\subsection{Spaces of test functions}\label{sec:tf}

We discuss transformability in terms of test functions in particular subclasses of $K(G):=C_c(G)$.
The above definition uses the function space
\begin{displaymath}
K_2(G)=\mbox{span} \{ f*\widetilde{f} : f \in C_c(G) \}.
\end{displaymath}
Note that $f*g\in K_2(G)$ for $f,g\in C_c(G)$, which follows from polarisation, see  \cite[Prop.~1.9.4]{MoSt} or \cite[Rem.~3.1.2]{P89}. The space $K_2(G)$ is dense in $C_c(G)$, which may be seen by approximating $f\in C_c(G)$ by convolution with a Dirac net and by polarisation, see p.~9 and Eq.~(4.6) in \cite{ARMA1} for the argument.

We are interested in $L^1$-characterisations of transformability. Note that if $\mu\in\mathcal M(G)$ is transformable, then for $f\in C_c(G)$ such that $\widehat f\in L^1(\widehat G)$ we even have $\widehat f\in L^1(\widehat\mu)$, see \cite[Prop.~3.1]{ARMA1}. Thus the space of functions
\begin{displaymath}
KL(G)= \{ f \in C_c(G) : \widehat{f} \in L^1(\widehat{G}) \}
\end{displaymath}
is important. As $C_c(G)\subset L^1(G)$, every $f\in KL(G)$ satisfies the inversion formula $f=\widecheck{\widehat f}$, see \cite[Thm.~3.5.8]{DE}.
Often, positive definiteness will be important for our arguments. For this reason, we will sometimes deal with the function space
\begin{displaymath}
PK(G):=\mbox{span}\{P(G)\cap C_c(G)\}.
\end{displaymath}
In fact elements of $PK(G)$ have an integrable Fourier transform.
\begin{lemma}\label{lm0}  $f \in P(G)\cap C_c(G)$ implies  $\widehat{f} \in L^1(\widehat{G})$. Hence $PK(G)=\mbox{span}\{P(G)\cap KL(G)\}$.
\end{lemma}

\begin{proof}
Let $f \in C_c(G)$ be positive definite.
Since $f$ is continuous and positive definite, by Bochner's theorem, see e.g.~\cite[Sec.~1.4.3]{RUD}, there exists a finite measure $\sigma\in \mathcal M(\widehat G)$ such that
\begin{displaymath}
f(x)= \int_{\widehat{G}} \chi(x) \, {\rm d} \sigma(\chi) \,.
\end{displaymath}
Then, by \cite[Thm.~2.2]{ARMA1}, \cite[Lemma~1.17]{MoSt}, the measure $f \cdot \theta_G$ is transformable and its Fourier transform is $\sigma$.
Again by  \cite[Thm.~2.2]{ARMA1} or \cite[Lemma~1.16]{MoSt}, as $f \in L^1(G)$, the measure $f \cdot \theta_G$ is transformable and its Fourier transform is $\widehat{f} \cdot \theta_{\widehat{G}}$.
Therefore, the uniqueness of the Fourier transform \cite[Thm.~2.1]{ARMA1} or \cite[Thm.~1.13]{MoSt} yields $\sigma= \widehat{f} \cdot \theta_{\widehat{G}}$.
As $\sigma$ is a  finite measure, it follows that $\widehat{f}\cdot \theta_{\widehat{G}}$ is a finite measure, and hence $\widehat{f} \in L^1(\widehat{G})$.
\end{proof}

\bigskip

We thus have the following relationship among our spaces:
\begin{displaymath}
K_2(G) \subset PK(G) \subset KL(G) \subset C_c(G) \,.
\end{displaymath}
As $K_2(G)$ is dense in $C_c(G)$, it follows that $PK(G)$ and $KL(G)$ are also dense in $C_c(G)$.
Many examples of functions in $PK(G)$ are provided by the following lemma.

\begin{lemma}\label{lm1} If $f ,g \in L^2(G)$ have compact support, then $f*g \in PK(G)$.
\end{lemma}

\begin{proof}
By polarisation, it suffices to prove the result in the case $g=\widetilde{f}$.
By \cite[Thm. on page 4, (d)]{RUD} or \cite[Lemma 3.4.1]{DE} we have $f * \widetilde{f} \in C_0(G)$. Moreover, as $f$ has compact support, so has $f* \widetilde{f}$.
This shows that $f *\widetilde{f} \in C_c(G)$.
By construction $f *\widetilde{f}$ is positive definite, which completes the claim.
\end{proof}

\begin{remark}\label{WKL}
For later use we note $1_W*\widetilde{1_W} \in PK(G) \subset KL(G)$ for relatively compact measurable $W\subset G$.
\end{remark}

We have the following characterisation of transformability by generalised PSF. For the Haar measure on a closed subgroup of $G$, part (ii) of the theorem is the classical PSF as in \cite[Thm.~5.5.2]{Rei2}.

\begin{prop}\label{charFT} For $\mu \in \mathcal M(G)$ and $\nu \in \mathcal M(\widehat{G})$ the following are equivalent:
\begin{itemize}
\item[(i)] $\mu$ is transformable and $\widehat{\mu}= \nu$.
\item[(ii)] For every $f\in KL(G)$  we have $\widecheck{f} \in L^1(\nu)$ and
$\langle \mu, f\rangle \,=\, \langle \nu, \widecheck{f} \rangle$.
\item[(iii)] For every $f\in PK(G)$  we have $\widecheck{f} \in L^1(\nu)$ and
$\langle \mu, f\rangle \,=\, \langle \nu, \widecheck{f} \rangle$.
\item[(iv)] For every $f\in K_2(G)$ we have $\widecheck{f} \in L^1(\nu)$ and
$\langle \mu, f\rangle \,=\, \langle \nu, \widecheck{f} \rangle$.
\end{itemize}
\end{prop}

\begin{proof}
(i) $\Rightarrow$ (ii) follows with \cite[Prop.~3.1]{ARMA1}.  (ii) $\Rightarrow$ (iii) $\Rightarrow$ (iv) holds since $K_2(G)\subset PK(G) \subset KL(G)$.  (iv) $\Rightarrow$ (i) holds trivially.
\end{proof}

\subsection{Double transformability}

For $f\in L^1(G)$ define $f^\dagger\in L^1(G)$ by $f^\dagger(x)=f(-x)$. Similarly, for $\mu\in\mathcal M(G)$ define $\mu^\dagger\in\mathcal M(G)$ via $\langle \mu^\dagger, f \rangle \,=\, \langle \mu, f^\dagger \rangle$ for all $f\in C_c(G)$.
If $\mu\in\mathcal M(G)$ and $\widehat\mu\in\mathcal M(\widehat G)$ are both transformable, then the inversion theorem $\widehat{\widehat\mu}=\mu^\dagger$  holds, and the measures $\mu,\mu^\dagger$ and $\widehat \mu$ are all translation bounded \cite[Thm.~3.4]{ARMA1}. This generalises the inversion theorem for integrable functions, see e.g.~\cite[Thm.~3.5.8]{DE}.
The mapping $\mu \mapsto \mu^\dagger$ can be seen as a natural extension  of the mapping $\mu \mapsto \widehat{\widehat\mu}$ from the subspace of twice Fourier transformable measures to the space of all measures. When considering this as a generalised double Fourier transform, one has to be careful as the extension doesn't preserve many properties which $\widehat{\widehat\mu}$ always has. For example, for a twice Fourier transformable measure $\widehat{\widehat{\mu}}$ is always translation bounded and weakly almost periodic \cite{ARMA,MoSt}, while for an arbitrary measure $\mu$, the measure $\mu^\dagger$ might not have these properties.

We now give a necessary and sufficient condition for a transformable measure to be twice transformable.

\begin{theorem}\label{double} Let $\mu\in\mathcal M(G)$ be transformable. Then the following are equivalent.

\begin{itemize}
\item[(i)] $\widehat\mu\in\mathcal M(\widehat G)$ is transformable.
\item[(ii)] For every $g \in K_2(\widehat{G})$ we have $\widecheck{g} \in L^1(\mu^\dagger)$.
\end{itemize}
If any of the above conditions holds, then $\widehat{\widehat\mu}=\mu^\dagger$,  and the measures $\mu, \mu^\dagger$ and $\widehat \mu$ are translation bounded.
\end{theorem}

\begin{proof}
``(i) $\Rightarrow$ (ii)'' Since both $\mu\in \mathcal M(G)$ and $\widehat \mu\in\mathcal M(\widehat G)$ are transformable, we have $\widehat{\widehat\mu}=\mu^\dagger\in\mathcal M(G)$. Now the claim follows from Proposition~\ref{charFT}.

\noindent ``(ii) $\Rightarrow$ (i)'' Consider any $g \in K_2(\widehat{G})$. Since by assumption $\widehat{g}\cdot \mu$ is a finite measure, its Fourier transform as a measure is the absolutely continuous measure with density function $I:\widehat G\to \mathbb C$ given by
\[
I(\chi):= \int_{G} \overline{\chi}(s) \widehat{g}(s) {\rm d} \mu (s)\,.
\]
Thus, for all $f \in K_2(G)$ we have $\langle \widehat{g}\cdot \mu , f \rangle = \langle   I \cdot \theta_{\widehat{G}}, \widecheck{f} \rangle$.
We also have $f\cdot \widehat{g} \in C_c(G)$ and $\widecheck{f \cdot\widehat{g}} = \widecheck{f}*g \in L^1(\widehat{G})$. Therefore, as $\mu$ is transformable we get
\[
\langle  \mu ,  f \cdot\widehat{g} \rangle = \langle \widehat{\mu},  \widecheck{f}*g \rangle = \langle (g^\dagger*\widehat{\mu})\cdot \theta_{\widehat{G}},  \widecheck{f} \rangle.
\]
Combining these we obtain
\begin{displaymath}
\langle I \cdot \theta_{\widehat{G}},  \widecheck{f} \rangle \, = \, \langle \widehat{g}\cdot \mu , f \rangle \,=\, \langle  (g^\dagger*\widehat{\mu})\cdot \theta_{\widehat{G}}, \widecheck{f} \rangle.
\end{displaymath}
This proves the measure equality $I\cdot \theta_{\widehat{G}} = (g^\dagger*\widehat{\mu}) \cdot \theta_{\widehat{G}}$. As $I$ and $g^\dagger*\widehat{\mu}$ are both continuous functions, they must be equal, and by evaluating at $0$ we get
$\langle  \mu , \widehat{g} \rangle = \langle  \widehat{\mu}, g \rangle$.
Therefore, for all $g \in K_2(\widehat{G})$ we have $\widecheck{g} \in L^1(\mu^\dagger)$ and $\langle  \widehat{\mu}, g \rangle = \langle \mu^\dagger, \widecheck{g} \rangle$. As $\mu^\dagger\in\mathcal M(G)$, the claim now follows from Proposition~\ref{charFT}.
\end{proof}

\begin{cor} Let $\mu \in \mathcal M(G)$ be transformable. If there exists some transformable $\nu \in \mathcal M(\widehat{G})$ such that $|\mu| \leq | \widehat{\nu} |$, then $\mu$ is twice transformable and $\mu\in\mathcal M^\infty(G)$.
\end{cor}
\begin{proof}
Let $g \in K_2(\widehat{G})$. As $\nu$ is Fourier transformable, we have by definition $\widehat{g} \in L^1( \widehat{\nu})$, and hence
$\left| \widehat{g} \right| \in L^1( | \widehat{\nu} |)$.
Since $\widehat{g} \in C_0(G)$ and $|\mu| \leq | \widehat{\nu} |$ we have
\[
\int_G \left| \widehat{g} \right| {\rm d} |\mu| \leq  \int_G \left| \widehat{g} \right| {\rm d} | \widehat{\nu} | < \infty \,.
\]
This shows that $\left| \widehat{g} \right| \in L^1( |\mu|)$ and hence $\widecheck{g} \in L^1( \mu^\dagger )$.
The claim follows now from Theorem \ref{double}, ``(ii) $\Rightarrow$ (i)''.
\end{proof}

In the case $G=\R^d$, we obtain a simpler version of Theorem~\ref{double}.

\begin{theorem}\label{double2} Let $\mu\in\mathcal M(\R^d)$ be transformable. Then $\widehat\mu\in\mathcal M(\widehat{\R^d})$ is transformable if and only if $\mu$ is translation bounded.
If any of the above conditions holds, then $\widehat{\widehat\mu}=\mu^\dagger$.
\end{theorem}

\begin{proof}
The implication ``$\Rightarrow$'' follows from Theorem~\ref{double}. For the implication ``$\Rightarrow$'',  let $g \in C_c(\widehat{\R^d})$. Then $g \in L^2(\widehat{\R^d})$ has compact support. Since $\mu^\dagger$ is translation bounded, we have $\widecheck{g} \in L^2(|\mu^\dagger|)$ by Theorem 1 in \cite{lin}.
It follows immediately that for all $f,g \in C_c(\widehat{\R^d})$ we have $\widecheck{f}\cdot \widecheck{g} \in L^1(|\mu^\dagger|)$. By taking linear combinations, we obtain that condition (ii) of Theorem~\ref{double} is satisfied. This proves the implication and the remaining claim.
\end{proof}

As a consequence we get
\begin{cor} Let $\mathcal M_{T}^\infty(\R^d)$ denote the space of Fourier transformable translation bounded measures on $\R^d$. Then, the Fourier transform is a bijection from $\mathcal M_{T}^\infty(\R^d)$ into itself. \qed
\end{cor}

\subsection{Fourier-Bohr coefficients}\label{FB coefficients}

 For this subsection $G$ will be a $\sigma$-compact LCA group. Then, for a transformable translation bounded measure $\mu$, the pure point part of $\widehat\mu$ can be computed by a certain averaging procedure.  To define suitable averaging sequences, consider for $U,W\subset G$ the \textit{(generalised) van Hove boundary}
\begin{displaymath}
\partial^UW=((U+\mathrm{cl}(W))\cap \mathrm{cl}(W^c))\cup ((U+\mathrm{cl}(W^c))\cap \mathrm{cl}(W)),
\end{displaymath}
which was introduced in \cite[Eqn.~(1.1)]{Martin2}, see also \cite[Sec.~2.2]{mr13} for a discussion. As $\partial W=\partial^{\{e\}}W\subset \partial^U W$ for $U$ any unit neighbourhood, the van Hove boundary may be considered as a thickened topological boundary in that case. A \textit{(generalised) van Hove sequence} is a sequence  $(A_n)_{n\in\mathbb N}$  of compact sets  in $G$ of positive (and finite) Haar measure, $0<\theta_G(A_n)<\infty$, such that for all compact $K\subset G$ we have
\begin{equation}\label{eq:vh}
\lim_{n\to\infty} \frac{\theta_G(\partial^K A_n)}{\theta_G(A_n)}=0.
\end{equation}
Existence of van Hove sequences in $G$ is discussed in \cite{Martin2}. In Euclidean space, any sequence of non-empty closed rectangular boxes of diverging inradius is a van Hove sequence. Also any sequence of non-empty closed balls of diverging radius is a van Hove sequence.

\begin{prop}[Fourier-Bohr coefficients]\label{prop:Hofgen}
For a $\sigma$-compact LCA group $G$, let $\mu\in\mathcal M^\infty(G)$ be transformable and consider $\chi\in\widehat G$. Let $(A_n)_{n\in\mathbb N}$ be any van Hove sequence in $G$. Then for every $t\in G$ we have
\begin{displaymath}
\widehat \mu(\{\chi\})=\lim_{n\to\infty} \frac{1}{\theta_G(A_n)} \int_{t+A_n} \overline{\chi}(x) \, {\rm d} \mu(x).
\end{displaymath}
The convergence is uniform in $t\in G$.
\end{prop}

\begin{remark}
The above average is sometimes called the \textit{Fourier-Bohr coefficient} of $\mu$ at $\chi$, compare \cite[Def.~2.1]{BD} or \cite[Eq. 8.14]{ARMA}.
The proposition extends Hof's result \cite[Thm.~3.2]{Hof1} to the non-Euclidean setting. The proposition also generalises a result of Lenz \cite[Cor.~5]{L09}, which has been derived for positive definite measures using dynamical systems. The proposition can partly be deduced from \cite[Thm.~11.3]{ARMA} when transformability of $\widehat \mu$ is granted. A statement about uniform convergence appears, in the Euclidean setting, in \cite[Lemma~2.4]{BD}. A short proof based on almost periodicity may be given as in  \cite[Thm.~1.19]{MoSt}.
\end{remark}

Our proof is an adaption of Hof's arguments based on the PSF. Hence the statement is a rather direct consequence of transformability. Let us first recall a fundamental property of characters.

\begin{lemma}\label{lem:fl}
Consider any $\sigma$-compact LCA group $G$ and let $\chi\in \widehat G$.
Then for every van Hove sequence $(A_n)_{n\in\mathbb N}$ in $G$ we have
\begin{displaymath}
\lim_{n\to\infty} \frac{1}{\theta_G(A_n)} \int_{A_n} \chi(x) \,{\rm d}\theta_G(x) =\delta_{\chi,e}.
\end{displaymath}
\end{lemma}

\begin{proof}
The conclusion of the lemma clearly holds for $\chi=e$. Consider any character $\chi\ne e$ and fix any $y\in G$ such that $\chi(y)\ne 1$. By translation invariance of the Haar measure on $G$ and $\chi(y+x)=\chi(y)\chi(x)$ we have
\begin{equation*}\label{eq:unif}
\int_{A_n} \chi(x) \,{\rm d}\theta_G(x)=
\int_G 1_{A_n}(y+x) \chi(y+x) \,{\rm d}\theta_G(x)=\chi(y) \int_{-y+A_n} \chi(x) \,{\rm d}\theta_G(x).
\end{equation*}
Due to the van Hove property of $(A_n)_{n\in\mathbb N}$, we have
\begin{displaymath}
\left|\int_{-y+A_n} \chi(x) \,{\rm d}\theta_G(x)- \int_{A_n} \chi(x) \,{\rm d}\theta_G(x)\right|
\le \theta_G((-y+A_n) \Delta A_n)\le\theta_G(\partial^{\{-y\}}A_n),
\end{displaymath}
which is $o(\theta_G(A_n))$ as $n\to\infty$ since $(A_n)_n$ is a van Hove sequence. Combining the above properties yields
\begin{displaymath}
|1-\chi(y)|\cdot\left|\frac{1}{\theta_G(A_n)} \int_{A_n} \chi(x) \,{\rm d}\theta_G(x)\right|
=o(1)
\end{displaymath}
as $n\to\infty$. Since $\chi(y)\ne1$, the statement of the lemma follows.
\end{proof}

\begin{proof}[Proof of Proposition~\ref{prop:Hofgen}]
We prove the proposition for $\chi=e$. The general case then follows from $\widehat \mu(\{\chi\})=(\delta_{\chi^{-1}}*\widehat \mu)(\{e\})$ and $(\delta_{\chi^{-1}} *\widehat\mu)=\widehat {\overline \chi \mu}$. We give an approximation argument using sufficiently smooth compactly supported functions.

Define $\varphi=\psi * \widetilde{\psi}\in C_c(G)$ for some $\psi\in C_c(G)$ such that $\int \psi\,{\rm d}\theta =1$. Then also $\int \varphi\,{\rm d}\theta=1$. Fix any van Hove sequence $(A_n)_{n\in\mathbb N}$ in $G$ and define
\[
f_n^t(x):=\frac{1}{\theta_G(A_n)} 1_{t+A_n}(x) \,.\]
We will also denote $(f_n^t)_\varphi:= \varphi * f_n^t$.
Then $(f_n^t)_\varphi\in KL(G)$ by Lemma~\ref{lm1}, and Proposition~\ref{charFT} (ii) yields $\widecheck{(f_n^t)_\varphi}\in L^1(\widehat\mu)$ and
\begin{equation}\label{eqn:PSFFB}
\langle \mu, (f_n^t)_\varphi \rangle \,=\, \langle \widehat{\mu}, \widecheck{(f_n^t)_\varphi} \rangle \,.
\end{equation}
We will identify the limit in the statement of the proposition from the rhs. The lhs will yield uniform convergence.
Let us start with the rhs. We first note $\widecheck{(f_n^t)_\varphi}(\chi)\to\delta_{\chi,e}$.
Indeed, $\widecheck{(f_n^t)_\varphi}(e)=\widecheck{ \varphi}(e) \cdot \widecheck{f_n^t}(e) =1$ since $\varphi$ is normalised. Furthermore for $\chi\ne e$ we have

\begin{displaymath}
\left|\widecheck{(f_n^t)_\varphi}(\chi)\right|= \left|\widecheck{\varphi}(\chi)\right|\cdot\left|\widecheck{f_n^t}(\chi)\right|
\le  || \varphi||_1\cdot \left|\widecheck{ f_n^t}(\chi)\right|  \to 0
\end{displaymath}
by Lemma~\ref{lem:fl}, as $(t+A_n)_{n\in\mathbb N}$ is a van Hove sequence.  Since $\varphi\in KL(G)$ by Lemma~\ref{lm1}, we have $\widecheck{\varphi}\in L^1(\widehat \mu)$ by \cite[Prop.~3.1]{ARMA1}.
In fact $\widecheck{\varphi}$ is an integrable majorant of $\widecheck{(f_n^t)_\varphi}$ as
\begin{displaymath}
\left|\widecheck{(f_n^t)_\varphi}\right|= \left|\widecheck{\varphi}\right|  \cdot \left|\widecheck{f_n^t}\right|
\le |\widecheck{\varphi}|\cdot ||f_n^t||_1 =|\widecheck{\varphi}|.
\end{displaymath}
We can thus apply Lebesgue's dominated convergence theorem for complex measures  \cite[Rem.~14.23]{HeRo} to obtain
\begin{displaymath}
\lim_{n\to\infty}
\langle \widehat{\mu},  \widecheck{(f_n^t)_\varphi} \rangle \,=\, \langle \widehat{\mu}, \delta_{\chi,e} \rangle =\widehat \mu(\{e\}).
\end{displaymath}
Now consider the limit $n\to\infty$ of the lhs of \eqref{eqn:PSFFB}. Since $f_n^t$ is proportional to the characteristic function of $t+A_n$, by a standard tedious computation which we omit, we have that  $f_n^t(x)\ne (f_n^t)_\varphi(x)$ implies $x\in t+\partial^K A_n$, where $K=\mathrm{supp}(\varphi)$. Hence we have for all $n$ the estimate
\begin{displaymath}
|\langle \mu, f_n^t \rangle - \langle \mu, (f_n^t)_\varphi \rangle|\le ||1-\varphi||_\infty\cdot \frac{|\mu|(t+\partial^K A_n)}{\theta_G(A_n)} \,.
\end{displaymath}
The rhs vanishes as $n\to\infty$ uniformly in $t\in G$, by translation boundedness of $\mu$ and by the van Hove property of $(A_n)_{n\in\mathbb N}$. This may be seen by inspecting the proof of \cite[Lemma~9.2 (b)]{LR}. Hence existence and uniformity of the limit follow.
\end{proof}

\section{Fourier analysis of weighted model sets}\label{sec:fapm}

Let $G,H$ be LCA groups and assume that $\cL$ is a lattice in $G\times H$. We will consider measures on $G$ supported on certain projected lattice subsets. The goal of this section is to prove that their Fourier transform formula follows from the PSF of the underlying lattice and vice versa.
In the case of a $\sigma$-compact LCA group $G$, a consequence is a certain averaging property of such measures which is also known as the density formula.

\subsection{Cut-and-project schemes and weighted model sets} We recall the definitions of a cut-and-project scheme and of a model set, compare \cite{BG2}.

\begin{defi}[Cut-and-project scheme]\label{def:cp}    Let $G,H$ be LCA groups, and let $\cL$ be a lattice in $G\times H$, i.e., a discrete co-compact subgroup of $G\times H$. We call $(G, H, \cL)$ a \emph{cut-and-project scheme} if, with canonical projections $\pi^G: G\times H\to G$, $\pi^H: G\times H \to H$,
\begin{itemize}

\item[(i)] the restriction $\pi^G|_{\cL}$ of $\pi^G$ to $\cL$ is one-to-one,

\item[(ii)] $\pi^H(\cL)$  is dense in H.
\end{itemize}
\end{defi}

\begin{remark}
Note that condition (ii) may be assumed to hold without loss of generality by passing from $H$ to $\overline{\pi^H(\cL)}$. Since $\pi^G|_{\cL}$ is one-to-one,  $\pi^G|_{\cL}$ is a group isomorphism between $\cL$ and $L=\pi^G(\cL)$, and hence invertible in that case. We thus get a map $\star: L \to H$ via the composition
$L \overset{(\pi^G|_{\cL})^{-1}}{\longrightarrow}  \cL \overset{\pi^H}{\longrightarrow} L^\star$, which is also called the star map.
It is readily checked  that $\cL= \{ (x,x^\star) : x \in L \}$. We will not use the star map in the sequel, since a number of our results hold without assumption (i).
\end{remark}

\begin{remark}
Assume that $(G, H, \cL)$ is a cut-and-project scheme. Then the annihilator $\cL_0 \subset \widehat{G} \times \widehat{H}$ of $\cL$ is a lattice, the lattice dual to $\cL$. Let us write $\pi^{\widehat G}: \widehat G\times \widehat H\to \widehat G$, $\pi^{\widehat H}: \widehat G\times \widehat H \to \widehat H$ for the canonical projections and define $L_0=\pi^{\widehat G}(\cL_0)$.
Pontryagin duality can be used to show that $\pi^{\widehat G}|_{\cL_0}$ is one-to-one/dense if and only if $\pi^H|_\cL$ is dense/one-to-one and $\pi^{\widehat H}|_{\cL_0}$ is one-to-one/dense if and only if $\pi^G|_{\cL}$ is dense/one-to-one, see e.g.~\cite[Sec.~5]{RVM3}.  Hence
$(\widehat{G}, \widehat{H},  \cL_0)$ is also a cut-and-project scheme, and we have a star map $\star: L_0\to\widehat H$ as in $(G,H,\cL)$. This dual cut-and-project scheme describes the diffraction of a model set from $(G,H,\cL)$.
\end{remark}

\begin{defi}[Model set] \rm  Let a cut-and-project scheme $(G, H, \cL)$ and a \textit{window} $W\subset H$ be given, where the latter is assumed to be relatively compact and measurable. Then $\oplam(W)=\pi^G(\cL\cap (G\times W))$ is called a \textit{weak model set}, compare \cite{RVM1, HR15}. If $W$ has non-empty interior, then $\oplam(W)$ is called a \textit{model set}.
We say that $\oplam(W)$ is a \emph{regular model set} if $W$ relatively compact, measurable, has non-empty interior and $\theta_H(\partial W)=0$.
\end{defi}

\begin{remark}[Delone sets]\label{rem:udrd}
Recall that $D\subset G$ is uniformly discrete if there exists a non-empty open set $U\subset G$ such that $x+U$ contains at most one point of $D$ for every $x\in G$. The set $D\subset G$ is called relatively dense in $G$ if there is a compact set $K\subset G$ such that $D+K=G$. If $D\subset G$ is both uniformly discrete and relatively dense, then $D$ is called a \textit{Delone set}.
If $W\subset H$ is relatively compact, then $\oplam(W)$ is uniformly discrete. This is a simple consequence of uniform discreteness of $\cL$.
If $W\subset H$ has non-empty interior, then $\oplam(W)$ is relatively dense. For details of the arguments see e.g.~\cite[Prop.~2.6~(i)]{RVM3}.
\end{remark}

\begin{defi}[Weighted model set] Let $(G,H,\cL)$ be a cut-and-project scheme. Any function $h : H \to \C$ is called a \emph{weight function} on $H$. Assume that $h$ is a weight function such that
\begin{displaymath}
\omega_h:= \sum_{(x,y)\in \cL} h(y) \delta_x
\end{displaymath}
is a measure on $G$. Then $\omega_h\in\mathcal M(G)$ is called a \emph{weighted model set from $(G,H,\cL)$}.
\end{defi}

\begin{remark} We will be interested in weight functions $h$ such that $\omega_h$ is a translation bounded or transformable measure. Any weak model set $\oplam(W)$ leads to the weighted model set
$\omega_h=\delta_{\oplam(W)}$ with $h=1_W$. In fact $\omega_h$ is then a translation bounded measure since $\oplam(W)$ is uniformly discrete, but a non-periodic $\omega_h$ may not be transformable, see the example \cite[p.~37]{Hof1}.
\end{remark}

An important class of weighted model sets arises from Riemann integrable weight functions. More generally, the following lemma holds.

\begin{lemma}\label{lem:tb} Let $h : H \to \C$ be a bounded and compactly supported weight function. Then $\omega_h$ is a weighted model set. In fact $\omega_h$ is a uniformly translation bounded measure on $G$.
\end{lemma}

\begin{remark}
We call  $\omega_h\in\mathcal M(G)$  \textit{uniformly translation bounded} if  the collection $\{\omega_{\delta_t*h}: t\in H\}$ is uniformly translation bounded,  where $(\delta_t*h)=h(\cdot-t)$. Keep in mind that this definition depends on the particular choice of $h$.
\end{remark}

\begin{proof} Since $h$ is bounded and of compact support, we find $0\le c<\infty$ and compact $W\subset H$ such that $|h|\le c\cdot1_W$. Fix arbitrary compact $K\subset G$ and $s\in G$, $t\in H$. We then have
\begin{displaymath}
\begin{split}
\left| \omega_{\delta_t*h}(s+K) \right| &=\left|\sum_{(x,y)\in \cL} (\delta_t*h)(y)1_{s+K}(x)\right| \leq\sum_{(x,y)\in \cL} c\cdot1_{t+W}(y)1_{s+K}(x)\\
&\le  c\cdot \sup_{(s,t)\in G\times H} \sharp (\cL\cap ((s+K)\times (t+W)))<\infty \,.
\end{split}
\end{displaymath}
By uniform discreteness of the lattice $\cL$, the term on the rhs is a finite constant which does not depend on $s\in G$ or $t\in H$. This shows that $\omega_h$ is a uniformly translation bounded measure and, in particular, a weighted model set.
\end{proof}

\subsection{Generalised PSF for weighted model sets}

The theorems of this subsection form the heart of this paper. We remark that conditions (i), (ii) of Definition~\ref{def:cp} are not used in the proofs. We first consider weighted model sets with positive definite weight functions. Note that part (i)$\Rightarrow$(ii) in the following theorem extends \cite[Lemma~9.3]{BG2} to the non-Euclidean setting.

\begin{theorem}[PSF for weighted model sets]\label{equiv PSF -diff-dens}

Let $(G,H,\cL)$ be a cut-and-project scheme with dual cut-and-project scheme $(\widehat G, \widehat H, \cL_0)$.
 Then the following are equivalent.

\begin{itemize}
\item[(i)] The lattice Dirac comb $\delta_{\cL}\in \mathcal M^\infty(G\times H)$ is transformable and satisfies the PSF
\begin{displaymath}
\widehat{\delta_{\cL}}=\mathrm{dens}(\cL) \cdot \delta_{\cL_0}.
\end{displaymath}
\item[(ii)]
For every $h\in KL(H)$, the weighted model set $\omega_h\in \mathcal M^\infty(G)$ is uniformly translation bounded, transformable and satisfies, with $\omega_{\widecheck h}\in \mathcal M^\infty(\widehat G)$,  the generalised PSF
\begin{displaymath}
\widehat{\omega_h}= \mathrm{dens}(\cL)\cdot \omega_{\widecheck{h}}\,.
\end{displaymath}

\item[(iii)]
For every $h\in K_2(H)$, the weighted model set $\omega_h\in \mathcal M^\infty(G)$ is uniformly translation bounded, transformable and satisfies, with $\omega_{\widecheck h}\in \mathcal M^\infty(\widehat G)$,  the generalised PSF
\begin{displaymath}
\widehat{\omega_h}= \mathrm{dens}(\cL)\cdot \omega_{\widecheck{h}}\,.
\end{displaymath}

\end{itemize}
\end{theorem}
The following remark serves as a preparation for the proof of Theorem~\ref{equiv PSF -diff-dens}.

\begin{remark}[test functions in product spaces]\label{rem:closure}
For given functions $g : G \to \C$ and $h : H \to \C$, we will denote by $g \odot h$ the function $G \times H \to \C$ given by $(g \odot h) (s,t)= g(s)\cdot h(t)$.
It is easy to see that $g \in C_c(G)$ and $h \in C_c(H)$ imply $g \odot h \in C_c(G \times H)$, and that $g \in K_2(G)$ and $h \in K_2(H)$ imply $g \odot h \in K_2(G \times H)$. Similarly, $g \in KL(G)$ and $h \in KL(H)$ imply $g \odot h \in KL(G \times H)$, and $g \in PK(G)$ and $h \in PK(H)$ imply $g \odot h \in PK(G \times H)$ by \cite[Lemma~3.2]{BF}.
\end{remark}

\begin{proof}[Proof of Theorem~\ref{equiv PSF -diff-dens} ``(i) $\Rightarrow$ (ii)'']

Fix any $h \in KL(H)$.  Then $\omega_h$ is a uniformly translation bounded measure since $\delta_{\cL}$ is translation bounded, see Lemma~\ref{lem:tb}.

We show that the linear functional $\omega_{\widecheck{h}}: f\mapsto \omega_{\widecheck{h}}(f)$ is a measure. Let $K \subset \widehat{G}$ be any compact set. By \cite[Lemma~3.4.5]{DE} 
we can find some $g'\in C_c(G)$ such that  $g= g'*\widetilde{ g'}$ satisfies $\widecheck{g}=|\widecheck{g'}|^2 \geq 1_K$.
In particular we have $g\in KL(G)$, which implies $g \odot h \in KL(G \times H)$ by Remark~\ref{rem:closure}. As a consequence we have  $\widecheck{g}\odot \widecheck{h}\in L^1(\delta_{\cL_0})$ by assumption and Theorem~\ref{PSFL-I}. In particular we have
%
 \begin{displaymath}
 \sum_{(\chi,\eta) \in \cL_0} |\widecheck{g}(\chi) \widecheck{h}(\eta)| =: C < \infty \,.
 \end{displaymath}
Let now $f \in C_c(\widehat{G})$ be such that $\supp(f) \subset K$. Then, as $|f| \leq \| f\|_\infty \cdot \widecheck{g}$ we have
 \begin{displaymath}
 \sum_{(\chi,\eta) \in \cL_0} \left|f(\chi)\right| |\widecheck{h}(\eta)|  \leq  \sum_{(\chi,\eta) \in \cL_0} \| f \|_\infty |\widecheck{g}(\chi)| |\widecheck{h}(\eta)|=C \| f \|_\infty  \,.
 \end{displaymath}
This shows that
 \begin{displaymath}
 \omega_{\widecheck{h}}(f):= \sum_{(\chi,\eta) \in \cL_0} f(\chi) \widecheck{h}(\eta)
 \end{displaymath}
is absolutely convergent, and therefore $\omega_{\widecheck{h}}(f)$ is well defined. Moreover, we have
 \begin{displaymath}
 \left| \omega_{\widecheck{h}}(f) \right| \leq \sum_{(\chi,\eta) \in \cL_0} \left|f(\chi)\right| |\widecheck h(\eta)|\leq C \| f \|_\infty  \,.
 \end{displaymath}
Therefore, for all $f \in C_c(\widehat{G})$ with $\supp(f) \subset K$ we have
 \begin{displaymath}
 \left| \omega_{\widecheck{h}}(f) \right| \leq C \| f \|_\infty  \,.
 \end{displaymath}
Since the constant $C$ depends only on $K$, it follows from the Riesz representation theorem that $\omega_{\widecheck{h}}$ is a measure.

Now consider arbitrary $g\in KL(G)$. Then $g \odot h \in KL(G \times H)$ by Remark~\ref{rem:closure}. By assumption and Theorem~\ref{PSFL-I} we have $\widecheck{g}\odot \widecheck{h}\in L^1(\delta_{\cL_0})$ and hence
\begin{displaymath}
\langle \delta_{\cL}, g\odot h \rangle=\mathrm{dens}(\cL)\cdot \langle \delta_{\cL_0}, \widecheck{g}\odot \widecheck{h} \rangle \in \mathbb C\,.
\end{displaymath}
By definition of $\omega_h$ we have $\langle \omega_h, g \rangle=\langle \delta_{\cL}, g\odot h \rangle$.  For the rhs of the above equation we note that by definition of the linear functional $\omega_{\widecheck{h}}$ we have
\begin{displaymath}
\langle \delta_{\cL_0}, \widecheck{g}\odot \widecheck{h} \rangle = \sum_{(\chi,\psi) \in \cL_0} \widecheck{h}(\psi)\cdot \widecheck{g}(\chi)=\langle\omega_{\widecheck{h}}, \widecheck{g}\rangle \,.
\end{displaymath}
Therefore we have for arbitrary $g \in KL(G)$ the equality
\begin{displaymath}
\langle \omega_h, g\rangle = \mathrm{dens}(\cL)\cdot\langle\omega_{\widecheck{h}}, \widecheck{g}\rangle \, .
\end{displaymath}
%
%
As $\cL$ is transformable, we have in particular $\widecheck g\in L^1(\omega_{\widecheck h})$ since $\widecheck g\odot \widecheck h\in L^1(\delta_{\cL_0})$. Hence $\omega_h$ is transformable by Proposition~\ref{charFT} (ii), with $\widehat{\omega_h}=\mathrm{dens}(\cL)\cdot \omega_{\widecheck{h}}$.
Translation boundedness of $\omega_{\widecheck{h}}$ follows from Remark~\ref{rem:FTprop}.

\end{proof}

As the statement  ``(ii) $\Rightarrow$ (iii)'' is trivial, it remains to show  ``(iii) $\Rightarrow$ (i)''.

\begin{proof}[Proof of Theorem~\ref{equiv PSF -diff-dens} ``(iii) $\Rightarrow$ (i)'']
The claim $\delta_\mathcal{L}\in\mathcal M^\infty(G\times H)$ follows from uniform translation boundedness of the measures $\omega_h$ for $h\in K_2(H)$ by elementary estimates.
Let $g\in K_2(G)$ and $h\in K_2(H)$. As before, from (iii) it is immediate that
\begin{equation}\label{help:eq}
\langle \delta_{\cL}, g\odot h \rangle=\mathrm{dens}(\cL)\cdot\langle \delta_{\cL_0}, \widecheck{g}\odot \widecheck{h} \rangle \,.
\end{equation}
By (iii) we have $\widecheck g \in L^1( \omega_{\widecheck h} )$, which means $\widecheck{g} \odot \widecheck{h} \in L^1(\delta_{\cL_0} )$.
Hence the PSF holds for all functions in $K_2(G)\odot K_2(H):=\{ g \odot h \,|\, g \in K_2(G), h \in K_2(H) \}$.
We split the proof for general functions from $K_2(G\times H)$ in four steps.

\emph{Step 1:} We show for any $g \in K_2(G)$ and $h \in K_2(H)$ that $\widehat{(g \odot h)}$ is convolvable with $\delta_{\cL_0}$ and that the function $\mathrm{dens}(\cL) \cdot\widehat{(g \odot h)} * \delta_{\cL_0}$ is the Fourier transform of the finite measure $(g \odot h)\cdot\delta_{\cL}$. In particular, this implies that $\widehat{(g \odot h)} * \delta_{\cL_0} \in C_U(\widehat{G}\times \widehat{H})$.

Let $(\chi, \psi) \in \widehat{G} \times \widehat{H}$. As $K_2(G)$ and $K_2(H)$ are closed under multiplication by characters, $(\overline{\chi} g) \odot (\overline{\psi} h) \in K_2(G)\odot K_2(H)$. Therefore, by the above we have $\widecheck{\overline{\chi} g} \odot \widecheck{\overline{\psi} h} \in L^1(\delta_{\cL_0} )$ and
$\langle \delta_{\cL}, (\overline{\chi} g) \odot (\overline{\psi} h)  \rangle=\mathrm{dens}(\cL)\cdot\langle \delta_{\cL_0}, \widecheck{\overline{\chi} g} \odot \widecheck{\overline{\psi} h} \rangle$. Therefore, we have
\begin{eqnarray*}
\begin{split}
\int_{G \times H} \overline{\psi(s) \chi(t)} \, & {\rm d} [(g \odot h)\cdot \delta_{\cL}](s,t) = \langle \delta_{\cL}, (\overline{\chi} g) \odot (\overline{\psi} h)  \rangle \\
&=\mathrm{dens}(\cL)\cdot\langle \delta_{\cL_0}, \widecheck{\overline{\chi} g} \odot \widecheck{\overline{\psi} h} \rangle =\mathrm{dens}(\cL)\cdot\langle \delta_{\cL_0}, T_{\chi, \psi} ( \widecheck{ g} \odot \widecheck{ h}) \rangle \\
&=\mathrm{dens}(\cL) \cdot (\widehat{g}\odot \widehat{h}) *  \delta_{\cL_0}(\chi, \psi),
\end{split}
\end{eqnarray*}
where the convolution makes sense as $T_{\chi, \psi} ( \widecheck{ g} \odot \widecheck{ h}) =\widecheck{\overline{\chi} g} \odot \widecheck{\overline{\psi} h} \in L^1(\delta_{\cL_0} )$.
This shows that the function $\mathrm{dens}(\cL) \cdot\widehat{(g \odot h)} * \delta_{\cL_0}$ is the Fourier transform of the finite measure $(g \odot h)\cdot \delta_{\cL}$.

\medskip

\emph{Step 2:} For any $g \in K_2(G)$, $h \in K_2(H)$ and  $f \in K_2(G \times H)$ we have $\widecheck{f} \in L^1( \widehat{(g \odot h)} * \delta_{\cL_0})$ and
\[
\langle \delta_{\cL} , f\cdot (g \odot h) \rangle =\mathrm{dens}(\cL) \cdot \langle  \widehat{(g \odot h)} * \delta_{\cL_0} \cdot \theta_{\widehat G\times \widehat H}, \widecheck{f}  \rangle.
\]
This follows immediately from $\langle  \delta_{\cL} , f\cdot(g \odot h)\rangle = \langle (g \odot h)\cdot\delta_{\cL}, f \rangle$ and from the fact that the Fourier transforms of a finite measure as a measure and as a finite measure coincide \cite[Thm.~2.2]{ARMA1}.

\medskip

\emph{Step 3:} We show that for any $g \in K_2(G)$, $h \in K_2(H)$ and $f \in K_2(G \times H)$, all positive definite, we have $\widecheck{(g \odot h)} * \widecheck{f} \in L^1(  \delta_{\cL_0})$ and
\[
\langle \widehat{(g \odot h)} * \delta_{\cL_0} \cdot \theta_{\widehat G\times \widehat H}, \widecheck{f}  \rangle = \langle  \delta_{\cL_0}, \widecheck{f} * \widecheck{(g \odot h)}  \rangle.
\]

By Step 2 we know that $\widecheck{f} \in L^1( \widehat{(g \odot h)} * \delta_{\cL_0})$, while by Step 1 we know that $\widehat{(g \odot h)} * \delta_{\cL_0}$ is given by the continuous function
\[
\widehat{(g \odot h)} * \delta_{\cL_0} (x) =\int_{\widehat{G} \times \widehat{H}} \widehat{(g \odot h)} (x-y) \, {\rm d} \delta_{\cL_0} (y).
\]
Therefore we have by positive definiteness
\[
0\le \langle \widehat{(g \odot h)} * \delta_{\cL_0} \cdot \theta_{\widehat G\times \widehat H}, \widecheck{f}  \rangle = \int_{\widehat{G} \times \widehat{H}}  \widecheck{f} (x) \int_{\widehat{G} \times \widehat{H}} \widehat{(g \odot h)} (x-y) {\rm d} \delta_{\cL_0} (y) \, {\rm d}\theta_{\widehat G\times \widehat H}(x) < \infty.
\]
By positivity, we can use Tonelli's theorem to exchange the order of integration. This results in
\begin{eqnarray*}
\begin{split}
\langle  \delta_{\cL_0} , \widecheck{f} * \widecheck{(g \odot h)}  \rangle&=\int_{\widehat{G} \times \widehat{H}} \widecheck{f} (x) * \widecheck{(g \odot h)} (x) \, {\rm d} \delta_{\cL_0} (y)  \\
&=\int_{\widehat{G} \times \widehat{H}}\int_{\widehat{G} \times \widehat{H}}  \widecheck{f} (x)  \widehat{(g \odot h)} (x-y) \, {\rm d}\theta_{\widehat G\times \widehat H}(x)  \, {\rm d} \delta_{\cL_0} (y)  \\
&= \int_{\widehat{G} \times \widehat{H}}  \widecheck{f} (x) \int_{\widehat{G} \times \widehat{H}} \widehat{(g \odot h)} (x-y) \, {\rm d} \delta_{\cL_0} (y) \,{\rm d}\theta_{\widehat G\times \widehat H}(x) < \infty,
\end{split}
\end{eqnarray*}
which proves Step 3.

\medskip

\emph{Step 4:} We prove the PSF for general functions in $K_2(G\times H)$. 
Consider without loss of generality that $f \in K_2(G \times H)$ is positive definite. As $f$ has compact support, we can find compact sets $G_0 \subset G, H_0 \subset H$ such that $\supp(f) \subset G_0 \times H_0$.
Next, pick two functions $g \in K_2(G)$ and $h \in K_2(H)$ such that $g \equiv1$ on $G_0$ and $h\equiv1$ on $H_0$. For example, one may pick a continuous function $g_1 \in C_c(G)$ with $\int_G g_1(t) {\rm d}\theta_G(t)=1$ and then some $g_2\in C_c(G)$ which is 1 on $G_0 -\supp(g_1)$. Then $g=g_1*g_2\in K_2(G)$ and $g\equiv1$ on $G_0$. We choose $h\in K_2(H)$ in the same way. Then by construction $g \odot h \equiv 1$ on $\supp(f)$, and we therefore have $f \cdot (g \odot h) =f$.
As both $\widecheck{f}$ and $\widecheck{g \odot h }$ are elements of $L^1(\widehat{G} \times \widehat{H}) \cap L^2(\widehat{G} \times \widehat{H})$, by taking the inverse Fourier transform we also get
\[
\widecheck{f} * \widecheck{(g \odot h)} =\widecheck{f}.
\]
In order to apply Step 3, we use depolarisation to find finitely many $a_i \in \C$, and finitely many positive definite $g_i \in  K_2(G), h_i \in K_2(H)$, such that $g \odot h = \sum_i a_i \cdot (g_i \odot h_i)$.
By Step 3 we have $\widecheck{(g_i \odot h_i)} * \widecheck{f} \in L^1(  \delta_{\cL_0})$  and
\[
\langle  \widehat{(g_i \odot h_i)} * \delta_{\cL_0} \cdot \theta_{\widehat G\times \widehat H}, \widecheck{f}  \rangle = \langle  \delta_{\cL_0}, \widecheck{f} * \widecheck{(g_i \odot h_i)}   \rangle.
\]
Therefore we have
\[
\widecheck{f}= \widecheck{f} * \widecheck{(g \odot h)} = \sum_{i} a_i\cdot \widecheck{f} * \widecheck{(g_i \odot h_i)} \in L^1(  \delta_{\cL_0})
\]
and
\begin{displaymath}
\begin{split}
\langle \delta_{\cL}, f \rangle &=\langle  \delta_{\cL}, f\cdot (g \odot h) \rangle =\mathrm{dens}(\cL) \cdot \langle  \widehat{(g \odot h)} * \delta_{\cL_0} \cdot \theta_{\widehat G\times \widehat H}, \widecheck{f}  \rangle\\
& =\mathrm{dens}(\cL) \cdot \sum_{i} a_i \cdot \langle  \widehat{(g_i \odot h_i)} * \delta_{\cL_0}\cdot \theta_{\widehat G\times \widehat H}, \widecheck{f}  \rangle\\ &=\mathrm{dens}(\cL) \cdot \sum_{i} a_i \cdot \langle   \delta_{\cL_0} , \widecheck{f}*\widecheck{(g_i \odot h_i)}   \rangle\\
& =\mathrm{dens}(\cL)\cdot \langle   \delta_{\cL_0} , \widecheck{f}*\widecheck{(g \odot h)}   \rangle =\mathrm{dens}(\cL)\cdot \langle   \delta_{\cL_0} , \widecheck{f}   \rangle.
\end{split}
\end{displaymath}
This completes the proof.
\end{proof}

The weighted Dirac combs $\omega_h$ in the previous theorem are in fact twice Fourier transformable.

\begin{theorem}[Double transformability for weighted model sets]\label{double-trans}
Let $(G,H,\cL)$ be a cut-and-project scheme with dual cut-and-project scheme $(\widehat G, \widehat H, \cL_0)$. Then for every $h\in KL(H)$, the weighted model set $\omega_h\in \mathcal M^\infty(G)$ is twice transformable, and $\omega_{\widecheck h}\in \mathcal M^\infty(\widehat G)$ satisfies the generalised PSF
\begin{displaymath}
\reallywidehat{\omega_{\widecheck{h}}}= \mathrm{dens}(\cL_0)\cdot \omega_{h^\dagger}.
\end{displaymath}
\end{theorem}
\begin{proof}
Recalling Theorem~\ref{equiv PSF -diff-dens}, it remains to be shown that $\omega_{\widecheck{h}}$ is transformable with transform $ \mathrm{dens}(\cL_0)\cdot\omega_{h^\dagger}$. Noting  $(\omega_h)^\dagger=\omega_{h^\dagger}$ and $\mathrm{dens}(\cL)\cdot\mathrm{dens}(\cL_0)=1$, by Theorem~\ref{double} it suffices to show that $\widecheck{g} \in L^1(\omega_{h^\dagger})$ for all $g \in K_2(\widehat{G})$.
As $h$ is compactly supported, there exists some $f \in C_c(\widehat H)$ such that $\left|\widecheck{f}\right|^2 \geq |h^\dagger|$. Let $ g \in K_2(\widehat{G})$. Then, as $g\odot( f*\widetilde{f}) \in K_2( \widehat{G} \times \widehat{H})$, by the PSF for $\cL_0$ we get $\widecheck{g} \odot \left|\widecheck{f}\right|^2 \in L^1(  \delta_{\cL})$ and
\[
0\le \mbox{dens}(\cL_0)\cdot \langle \delta_{\cL} , \widecheck{g}\odot \left|\widecheck{f}\right|^2 \rangle = \langle \delta_{\cL_0}, g\odot (f*\widetilde{f}) \rangle <\infty \,.
\]
Using $\left|\widecheck{f}\right|^2 \geq |h^\dagger|$ we thus get $\sum_{(x,y) \in \cL} \left| \widecheck{g}(x) \right||h^\dagger(y)| < \infty$, which means $\widecheck{g} \in L^1(\omega_{h^\dagger})$.
\end{proof}

\subsection{Density formula for weighted model sets}

A consequence of the lattice PSF is a certain averaging property which is known as the density formula for regular model sets. See~\cite[Sec.~3]{HR15} for a discussion of its history.  Note that condition (i) of Definition~\ref{def:cp} is not used in the following proofs.

\begin{prop}[Density formula for weight functions in $PK(H)$]\label{density general lattices}  Let $(G,H,\cL)$ be a cut-and-project scheme with $\sigma$-compact $G$, and let $(A_n)_{n\in\mathbb N}$ be any van Hove sequence in $G$. Then for all $h \in PK(H)$ and for all $t\in G$ we have
\begin{displaymath}
\lim_{n\to\infty} \frac{ \omega_h(t+A_n)}{\theta_G(A_n)}=\mathrm{dens}(\mathcal \cL)\cdot \int_H h \, {\rm d}\theta_H
\end{displaymath}
The convergence is uniform in $t\in G$.
\end{prop}

\begin{proof}
This follows from the generalised PSF Theorem~\ref{equiv PSF -diff-dens}.
As $\omega_{\widecheck{h}}$ is a measure, we have $\omega_{\widecheck{h}}(\{ e \}) =  \widecheck{h} (e)$. Here we used that $\pi^{\widehat G}|_{\cL_0}$ is one-to-one, which follows from denseness of $\pi^H(\cL)$ in $H$ by Pontryagin duality.
Moreover, as $\omega_{h}$ is a translation bounded measure by Lemma~\ref{lem:tb} and transformable, we can apply Proposition~\ref{prop:Hofgen} to obtain
\begin{displaymath}
\widehat{\omega_h}(\{ e\}) = \lim_{n\to\infty} \frac{ \omega_h(t+A_n)}{\theta_G(A_n)}
\end{displaymath}
uniformly in $t\in G$. The claim follows now from Theorem~\ref{equiv PSF -diff-dens}.
\end{proof}

The range of the density formula can be extended to Riemann integrable weight functions $h:H\to\mathbb C$ by a standard approximation argument, see e.g.~ \cite{BM}. For the convenience of the reader, we repeat the short argument.

\begin{theorem}[Density formula for Riemann integrable weight functions]\label{th:denfin2}
Let $(G,H,\cL)$ be a cut-and-project scheme with $\sigma$-compact $G$.  If $h : H \to \mathbb C$ is Riemann integrable, then for every van Hove sequence $(A_n)_{n\in\mathbb N}$ in $G$ the density formula holds, i.e., for every $t\in G$ we have
\begin{displaymath}
\lim_{n\to\infty} \frac{\omega_h(t+A_n)}{\theta_G(A_n)}=\mathrm{dens}(\cL)\cdot \int_H h \,{\rm d}\theta_H.
\end{displaymath}
The convergence is uniform in $t\in G$.
\end{theorem}

\begin{proof}
We assume without loss of generality that $h$ is real valued.
Let $\varepsilon >0$ and define $c=\mathrm{dens}(\cL)>0$. Since $h$ is Riemann integrable, by the density of $K_2(H)$ in $C_c(H)$, there exists two functions $g_1,g_2 \in K_2(G)$ such that $g_1 \leq h \leq g_2 $ and $\int (g_2 -g_1)\, {\rm d}\theta_H  \le \frac{\varepsilon}{2c}$.
By the density formula Proposition~\ref{density general lattices}
there exists an $N$ such that for all $n \ge N$, all $t \in G$ and $i\in\{1,2\}$ we have
\begin{displaymath}
\left| \frac{\omega_{g_i}(t+A_n)}{\theta_G(A_n)}-c \int g_i \,{\rm d}\theta_H \right| \le \frac{\varepsilon}{2}.
\end{displaymath}
Thus, as $\omega_{g_1} \leq \omega_h \leq \omega_{g_2}$, for all $n \geq N$ and all $t \in G$ we have
\begin{eqnarray*}
\begin{split}
&c \int h  \, {\rm d}\theta_H - \frac{\omega_{h}(t+A_n)}{\theta_G(A_n)}  \leq
\frac{\varepsilon}{2}+c \int g_2  \, {\rm d}\theta_H - \frac{\omega_{g_1}(t+A_n)}{\theta_G(A_n)} \le \varepsilon
\end{split}
\end{eqnarray*}
and similarly
\begin{eqnarray*}
\begin{split}
&c \int h  \, {\rm d}\theta_H - \frac{\omega_{h}(t+A_n)}{\theta_G(A_n)}  \ge
-\frac{\varepsilon}{2}+c \int g_1  \, {\rm d}\theta_H - \frac{\omega_{g_2}(t+A_n)}{\theta_G(A_n)} \ge -\varepsilon
\end{split}
\end{eqnarray*}
Hence the claim of the theorem follows.
\end{proof}

\section{Diffraction of weighted model sets}\label{sec:pcp}

\subsection{Autocorrelation of weighted model sets}\label{sec:ac}

The following result is well-known, see e.g. \cite{BM} and \cite[Sec.~9.4]{BG2}. For the convenience of the reader, we revisit its proof and note that condition (i) of Definition~\ref{def:cp} does not enter in the arguments.

\begin{prop}\label{thm:ac}
Let $(G,H,\cL)$ be a cut-and-project scheme with $\sigma$-compact $G$. Let $h:H\to\mathbb C$ be Riemann integrable. Then the weighted model set $\omega_h\in\mathcal M^\infty(G)$ has a unique autocorrelation measure $\gamma=\omega_h\circledast \widetilde{\omega_h}\in\mathcal{M}^\infty(G)$ which is given by
\begin{displaymath}
\gamma=\mathrm{dens}(\cL)\cdot\omega_{h*\widetilde h} \,.
\end{displaymath}
\end{prop}

\begin{proof}
Fix any van Hove sequence $(A_n)_{n\in\mathbb N}$ in $G$. According to Section~\ref{sec:ele}, the autocorrelation of $\omega_h\in\mathcal M^{\infty}(G)$ is defined as the vague limit of the finite autocorrelation measures $\gamma_n$ given by
 \begin{displaymath}
\gamma_n= \frac{1}{\theta_G(A_n)} \, \omega_h|_{A_n}*\widetilde{\omega_h|_{A_n}}
= \frac{1}{\theta_G(A_n)}\, \omega_h|_{A_n}*\omega_{\widetilde h}|_{-A_n}
=\sum_{(z, z')\in \cL} \eta_n'(z')\delta_z,
\end{displaymath}
where $|_{A_n}$ denotes restriction to $A_n$, and where $\eta_n'(z')$ is given by
 \begin{displaymath}
\eta_n'(z')=  \frac{1}{\theta_G(A_n)} \sum_{(x,x') \in \cL \cap (A_n\cap(z+A_n)\times H)} h(x')\overline{h(x'-z')}.
\end{displaymath}
For fixed $n$, the above sum is finite since $\omega_h$ is a measure and $h$ is bounded. Also note
 \begin{displaymath}
\left|\sum_{(x,x')\in \cL \cap (A_n\Delta(z+A_n))\times H} h(x')\overline{h(x'-z')}\right|\le \|h\|_\infty \cdot |\omega_h|(\partial^{\{z\}} A_n)=o(\theta_G(A_n))
\end{displaymath}
as $n\to\infty$ since $\omega_h$ is translation bounded and $A_n$ is a van Hove sequence, see \cite[Lemma~9.2~(b)]{LR}. This shows that vaguely $\gamma_n=\sum_{(z,z')\in \cL} \eta_n(z') \delta_z+o(1)$ as $n\to\infty$ where
\begin{displaymath}
\eta_n(z')= \frac{1}{\theta_G(A_n)} \sum_{(x,x')\in \cL\cap A_n\times H} h(x')\overline{h(x'-z')}.
\end{displaymath}
Since the function $y\mapsto h(y)\overline{h(y-z')}$ is Riemann integrable on $H$, we can apply the density formula Theorem~\ref{th:denfin2} and obtain
\begin{displaymath}
\begin{split}
\eta(z')&=\lim_{n\to\infty} \eta_n(z')=\mathrm{dens}(\cL)\cdot \int_H h(y)\overline{h(y-z')}\,{\rm d}\theta_H(y)\\
&=\mathrm{dens}(\cL)\cdot (h*\widetilde{h})(z').
\end{split}
\end{displaymath}
Since $\omega_{h*\widetilde h}$ is uniformly discrete, this implies that $\gamma_n$ converges vaguely to $\gamma$, and the claim follows.
\end{proof}

\subsection{Transformability of weighted model sets}

The previous results can be used to characterise transformability of a weighted model set with continuous compactly supported weight function. The following Theorem~\ref{thm:charKL} emphazises the role of the function space $KL(H)$.

We start with a simple Lemma, which is inspired by \cite[Sect.~6.1]{RicStr2}. Recall that $C_0(H)$ denotes the space of continuous functions on $H$ vanishing at infinity.


\begin{lemma}\label{tb implies L1} Let $(G,H,\cL)$ be a cut-and-project scheme and let $h \in C_0(H)$. If $\omega_h \in {\mathcal M}^\infty(G)$ then $h \in L^1(H)$.
\end{lemma}
\begin{proof}
Since $\omega_h \in {\mathcal M}^\infty(G)$ we have $\omega_{|h|} = |\omega_h| \in  {\mathcal M}^\infty(G)$.
Fix some non-negative $f \in C_c(G)$ which satisfies $\int_G f(t) {\rm d}\theta_G(t) =1$. By translation boundedness we then have
\begin{displaymath}
\| \omega_{|h|}*f \|_\infty =: C < \infty \,.
\end{displaymath}
Now, assume by contradiction that $h \notin L^1(H)$. Then $\int_{H} |h(t)| \,{\rm d}\theta_H(t) =\infty$, which means that $\left| h \right| \cdot \theta_H$ is an infinite positive Radon measure. Therefore, by outer regularity
%
%
there exists a compact set $W \subset H$ such that
\begin{displaymath}
\left(\left| h \right| \cdot \theta_H\right)(W) > \frac{C+2}{\dens(\cL)} \,,
\end{displaymath}
and by choosing any non-negative $g \in C_c(H)$ such that $g \geq 1_W$ we have
\begin{displaymath}
 \dens(\cL)\cdot\int_H g(t) \, {\rm d}\theta_H(t)  \geq C+2 \,. \\
\end{displaymath}
Now $\omega_g$ possesses a density by Theorem~\ref{th:denfin2}. Fixing some van Hove sequence $(A_n)_{n\in \mathbb N}$ in $G$, we thus have
\begin{displaymath}
\lim_{n\to\infty} \frac{\omega_g(A_n)}{\theta_G(A_n)}= \dens(\cL) \cdot \int_H g(t) \, {\rm d}\theta_H (t) \geq C+2 \,.
\end{displaymath}
In particular, there exists some $N$ so that for all $n>N$ we have
\begin{displaymath}
\frac{\omega_g(A_n)}{\theta_G(A_n)} \ge  C+1 \,.
\end{displaymath}
Now, fix a compact set $K$ such that $\supp(f) \subset K$. Then $(\omega_g)|_{A_n}*f$ is zero outside $A_n+K$. Since $\omega_g$ and $f$ are non-negative, it is immediate to check that
\begin{displaymath}
(\omega_g)|_{A_n}*f \leq (\omega_g*f)|_{A_n+K} \,.
\end{displaymath}
Therefore, by Tonelli's theorem we have
\begin{eqnarray*}
\omega_g(A_n) &=& \int_G \left(\int_{G} f(t-s)\, {\rm d}\theta_G(t) \right) {\rm d} (\omega_g)|_{A_n} (s) \\
   &=& \int_G \left(\int_{G} f(t-s) \, {\rm d} (\omega_g)|_{A_n} (s) \right) {\rm d}\theta_G(t)  \\
   &=& \int_G \left((\omega_g)|_{A_n}*f\right) (t) \, {\rm d}\theta_G(t) \\
   &\leq& \int_G (\omega_g*f)|_{A_n+K}(t)\, {\rm d}\theta_G(t) =\int_{A_n+K} (\omega_g*f)(t) \, {\rm  d}\theta_G(t)\\
   &\leq& \int_{A_n+K} \|\omega_g*f\|_\infty \, {\rm d}\theta_G(t) =  \|\omega_g*f\|_\infty\cdot \theta_G(A_n+K) \, .\\
\end{eqnarray*}
It follows that for all $n >N$ we have
\begin{displaymath}
C+1 \leq \|\omega_g*f\|_\infty\cdot  \frac{\theta_G(A_n+K)}{\theta_G(A_n)} \,.
\end{displaymath}
Next, since $0 \leq \omega_{g} \leq \omega_{|h|}$ and $f \geq 0$ we have $0 \leq \omega_{g}*f \leq \omega_{|h|}*f$ and hence
\begin{displaymath}
C+1 \leq \|\omega_{|h|}*f\|_\infty\cdot \frac{\theta_G(A_n+K)}{\theta_G(A_n)} = C \cdot \frac{\theta_G(A_n+K)}{\theta_G(A_n)}
\end{displaymath}
for all $n >N$. But this is a contradiction, since by the van Hove property we have
\begin{displaymath}
\lim_{n \to \infty}  \frac{\theta_G(A_n+K)}{\theta_G(A_n)} =1 \,.
\end{displaymath}
\end{proof}

In conjunction with Theorem~\ref{equiv PSF -diff-dens}, the above result can be used to prove the following characterisation of transformability. In addition, we use almost periodicity and Bombieri-Taylor type results.

\begin{theorem}\label{thm:charKL} Let $(G,H, \cL)$ be a cut-and-project scheme with $\sigma$-compact $G$,  and let $h \in C_c(H)$. Then $\omega_h$ is Fourier transformable if and only if $\widecheck{h} \in L^1(\widehat{H})$. Moreover, in this case we have
\begin{displaymath}
\widehat{\omega_h} = \mathrm{dens}(\cL) \cdot \omega_{\widecheck{h}} \,.
\end{displaymath}
\end{theorem}
\begin{proof}
``$\Rightarrow$'': Since $h \in C_c(H)$, the weighted Dirac comb $\omega_h$ is a strongly almost periodic measure, see e.g.~\cite{BM,LR,NS11}. Therefore, as $\omega_h$ is assumed to be transformable, its Fourier transform $\widehat{\omega_h}$ is a pure point measure. Fix a van Hove sequence $(A_n)_{n\in\mathbb N}$ in $G$. Then by Proposition~\ref{prop:Hofgen} we have
\begin{displaymath}
\widehat{\omega_h}(\{ \chi \})
=\lim_{n\to\infty} \frac{(\overline{\chi} \cdot \omega_h)(t+A_n)}{\theta_G(A_n)}\,.
\end{displaymath}
uniformly in $t\in G$.
If $(\chi, \chi^\star) \in \cL_0$ then it is easy to check that $\overline{ \chi}\cdot\omega_h=\omega_{\overline{\chi^\star} \cdot h}$,  and therefore by the density formula for Riemann integrable weight functions Theorem~\ref{th:denfin2} we have
\begin{equation}\label{eq 22}
\widehat{\omega_h}(\{ \chi \})=\dens(\cL) \cdot \widecheck{h}(\chi^\star) \, .
\end{equation}
By Proposition~\ref{thm:ac}, the measure $\omega_h\in\mathcal M^\infty(G)$ has a unique autocorrelation measure
$\gamma=\mathrm{dens}(\cL)\cdot\omega_{h*\widetilde h}$,
and we have
\begin{equation}\label{eq 23}
\widehat{\gamma}(\{ \chi \})=\left| \widehat{\omega_h}(\{ \chi \}) \right|^2 \,,
\end{equation}
see e.g.~\cite[Thm.~3.4]{Hof1}, \cite[Thm.~5(c)]{L09}, or \cite[Thm.~4.3]{LS2}.
If $\chi \notin \pi^{\widehat G}(\cL_0)$, we thus have $\widehat{\gamma}(\{ \chi \})=0=\widehat{\omega_h}(\{ \chi \})$.
%
%
%
Thus by combining (\ref{eq 22}) with (\ref{eq 23}) we get
\begin{displaymath}
\widehat{\omega_h}=\dens(\cL) \cdot  \omega_{\widecheck{h}} \,.
\end{displaymath}
Since the measure $\omega_h$ is Fourier transformable, its Fourier transform $\widehat{\omega_h}$ is a translation bounded measure, compare Remark \ref{rem:FTprop}. Therefore, by Lemma~\ref{tb implies L1} we have $\widecheck{h} \in L^1(\widehat{H})$.

``$\Leftarrow$'': Since $h \in C_c(H)$ is assumed to satisfy  $\widecheck{h} \in L^1(\widehat{H})$, the claim is the statement in Theorem~\ref{equiv PSF -diff-dens} (ii).
\end{proof}

\subsection{Pure point diffraction in regular model sets}

The following theorem is our main result.

\begin{theorem}[Pure point diffraction of weighted model sets]\label{theo:main}

Let $(G, H, \cL)$ be a cut-and-project scheme with $\sigma$-compact $G$ and denote by $(\widehat G, \widehat H, \cL_0)$ its dual. Then the following are equivalent.
\begin{itemize}
\item[(i)] The lattice Dirac comb $\delta_{\cL}\in \mathcal M^\infty(G\times H)$ is transformable and satisfies the PSF
\begin{displaymath}
\widehat{\delta_{\cL}}=\mathrm{dens}(\cL) \cdot \delta_{\cL_0}\, .
\end{displaymath}
\item[(ii)] The lattice Dirac comb $\delta_{\cL}\in \mathcal M^\infty(G\times H)$  has autocorrelation $\gamma\in\mathcal M^\infty(G\times H)$  and diffraction $\widehat \gamma\in\mathcal M^\infty(\widehat G\times \widehat H)$ given by
\begin{displaymath}
\gamma=\mathrm{dens}(\cL)\cdot \delta_{\cL}, \qquad \widehat \gamma=\mathrm{dens}(\cL)^2\cdot \delta_{\cL_0}\, .
\end{displaymath}
\item[(iii)] For every Riemann integrable function $h:H\to\mathbb C$, the weighted model set $\omega_h\in\mathcal M^\infty(G)$ is uniformly translation bounded, with  autocorrelation $\gamma\in\mathcal M^\infty(G)$  and diffraction $\widehat \gamma\in\mathcal M^\infty(\widehat G)$ given by
\begin{displaymath}
\gamma=\mathrm{dens}(\cL)\cdot \omega_{h*\widetilde{h}}, \qquad \widehat \gamma=\mathrm{dens}(\cL)^2\cdot\omega_{|\widecheck{h}|^2}\, .
\end{displaymath}
\end{itemize}
In particular, $\omega_h$ has pure point diffraction for every Riemann integrable function $h:H\to\mathbb C$.
\end{theorem}

\begin{remark}
The implication ``(i) $\Rightarrow$ (iii)'' is the well-known diffraction formula as in \cite{Hof1, Martin2, BM}.
Part (iii) of the above theorem applies to regular model sets, as for any relatively compact measurable $W\subset H$ such that $\theta_H(\partial W)=0$, its characteristic function $h=1_W:H\to\mathbb R$ is Riemann integrable.
\end{remark}

\begin{proof}[Proof of Theorem~\ref{theo:main}]
 ``(i) $\Rightarrow$ (ii)'' follows if the autocorrelation formula has been established. But the expression for $\gamma$ is a special case of Proposition~\ref{thm:ac} with trivial internal space and $\omega_h$ a lattice Dirac comb. The reverse implication  ``(ii) $\Rightarrow$ (i)'' is trivial.

\noindent ``(i) $\Rightarrow$ (iii)'' Uniform translation boundedness of $\omega_h$ is Lemma~\ref{lem:tb}. The explicit form of the autocorrelation $\gamma$ is Proposition~\ref{thm:ac}, which ultimately relies on Theorem~\ref{equiv PSF -diff-dens} (i) $\Rightarrow$ (ii). The statement about $\widehat\gamma$ now follows from  Theorem~\ref{equiv PSF -diff-dens} (i) $\Rightarrow$ (ii).

The implication ``(iii) $\Rightarrow$ (i)'' is  Theorem~\ref{equiv PSF -diff-dens} (iii) $\Rightarrow$ (i) applied to the autocorrelation measure,  as any function in $K_2(H)$ is Riemann integrable.
\end{proof}

\begin{remark}[Modified Wiener diagram]\label{rem:mwd}

For any weighted model set $\omega_h$ with weight function $h\in KL(H)$ the modified Wiener diagram
\begin{displaymath}
\begin{CD}
\omega_h @>\circledast>>\mathrm{dens}(\cL) \cdot \omega_{h*\widetilde h}\\
@V{\mathcal F}VV @VV{\mathcal F}V\\
\mathrm{dens}(\cL)\cdot \omega_{\widecheck h} @>{|\cdot|^2}>> \mathrm{dens}(\cL)^2\cdot \omega_{|\widecheck h|^2}
\end{CD}
\end{displaymath}
commutes, as $\omega_{\widecheck h}$ is a measure in that case by Theorem~\ref{equiv PSF -diff-dens}. This includes lattice Dirac combs, as one may choose $H$ trivial in that case. The diagram may no longer commute for a general Riemann integrable weight function $h$, as $\omega_{\widecheck h}$ might not be a measure in that case, and as $\omega_h$ might not be a transformable measure. However the upper right path is still well defined in that case, such that the diffraction of $\omega_h$ is a pure point measure and may be computed by ``squaring the Fourier-Bohr coefficients''.
\end{remark}

\subsection{Double transformability for measures with Meyer set support.}
Next, let $\Lambda\subset G$ be a Meyer set  \cite[Def.~7.2]{NS11}, i.e., $\Lambda$ and $\Lambda-\Lambda-\Lambda$ are both uniformly discrete and relatively dense.
Let $\mu$ be a measure supported inside $\Lambda$. Note that $\mu$ does not need to be pure point diffractive. We prove that if $\mu$ is transformable, then it is automatically twice transformable. As any model set is a Meyer set, this extends Theorem~\ref{double-trans}.
\begin{theorem}\label{Twice FT} Let $G$ be $\sigma$-compact, let $\Lambda\subset G$ be a Meyer set and let $\mu\in\mathcal M^\infty(G)$ be supported inside $\Lambda$. If $\mu$ is transformable, then $\widehat{\mu}$ is also transformable, and
we have $\widehat{\widehat{\mu}} = \mu^\dagger$.
\end{theorem}
\begin{proof}
We check the integrability condition in Theorem~\ref{double}. Since $\Lambda$ is a Meyer set, there exists a cut-and-project scheme $(G, H, \mathcal{L})$ and a window $W$ such that $\Lambda \subset \oplam(W)$ by \cite[Thm.~1.8]{NS11}.
Also, as $\mu$ is translation bounded, there exists some finite positive constant $c$ such that
$\left| \mu(\{ x \} ) \right| \leq c$ for all $x \in \Lambda$.
Next, we pick some $h \in K_2(H)$ such that $h \geq c \cdot 1_W$. Then
$\left| \mu \right| \leq c\cdot\omega_h$. Now, by Theorem~\ref{double-trans}, the measure $\omega_h$ is twice transformable. Therefore, for all $g \in K_2(\widehat{G})$ we have $\widecheck{g} \in L^1(\omega_{h^\dagger})$.
Hence, as $\left| \mu^\dagger \right| \leq c\cdot \omega_{h^\dagger}$ we get $\widecheck{g} \in L^1(\mu^\dagger)$.
\end{proof}
\begin{remark}
Our validation of the integrability condition in Theorem~\ref{double} relies on the lattice PSF: we embed the Meyer set into a regular model set coming from a cut-and-project scheme $(G, H, \cL)$, and then the integrability condition follows from the PSF applied to the dual lattice $\cL_0$.
\end{remark}

\subsection{Pure point diffraction in weak model sets}

 We complete the paper by looking in the next two subsections to some recent results about weak model sets and Meyer sets, and their connection to the PSF. In this subsection we look at the pure point diffractivity of weak model sets satisfying a certain density condition, which has been proven in \cite{Nicu14,KR}.

Consider a weak model set $\oplam(W)$. By definition, the window $W\subset H$ is relatively compact and measurable. Hence $\omega_h$ where $h=1_W$ is a weighted model set in that case by  Lemma~\ref{lem:tb}. But $\omega_h$ may not be pure point diffractive. On the other hand $h*\widetilde{h} \in PK(G)$ by Lemma \ref{lm1}, which means that the measure $\omega_{h*\widetilde h}$ is transformable by Theorem~\ref{theo:main}. Thus the question arises which weighted model sets $\omega_h$ have an autocorrelation given by $\mathrm{dens}(\cL)\cdot \omega_{h*\widetilde h}$. As argued by Moody \cite{RVM1}, this property is typical when one considers the ensemble of weak model sets with all shifts of a given window together with the uniform measure on this ensemble \cite[Theorem~1]{RVM1}. For compact windows, it is related to maximal density of the weak model set \cite[Prop.~3.4]{HR15}.  The following result extends \cite[Cor.~1]{RVM1}.

\begin{theorem}\cite[Thm.~7]{Nicu14}
Let $(G,H,\cL)$ be a cut-and-project scheme with $\sigma$-compact $G$ and let $\omega_h$ be the Dirac comb of a weak model set, i.e., $h=1_W$ for some relatively compact measurable $W\subset H$. Assume that there exists a van Hove sequence $(A_n)_{n\in\mathbb N}$ in $G$ such that $\omega_h$ has maximal density with respect to $(A_n)_n$, i.e.,
\begin{displaymath}
\lim_{n\to\infty} \frac{1}{\theta_G(A_n)} \omega_h(A_n)=\mathrm{dens}(\cL) \cdot \theta_H(\overline{W}).
\end{displaymath}
Then, with respect to the given van Hove sequence $(A_n)_n$, the weak model set $\omega_h$ has autocorrelation $\gamma$ and diffraction $\widehat \gamma$ given by
\begin{displaymath}
\gamma=\mathrm{dens}(\cL)\cdot \omega_{g*\widetilde{g}}, \qquad \widehat \gamma=\mathrm{dens}(\cL)^2\cdot\omega_{|\widecheck{g}|^2},
\end{displaymath}
where $g=1_{\overline{W}}$. In particular, $\omega_h$ has pure point diffraction. \qed
\end{theorem}

 The proof of the equality $\gamma=\mathrm{dens}(\cL)\cdot \omega_{g*\widetilde{g}}$ in \cite{Nicu14} is done by a computation which is similar in idea but more technical than the proof of Theorem~\ref{thm:ac}. The diffraction formula follows then from Theorem \ref{double-trans}, and hence can be seen as a consequence of PSF.

\begin{remark}
The above result reduces to the diffraction formula for regular model sets, since any regular model set has maximal density  \cite[Prop.~3.4]{HR15}. For a model set with window satisfying $\theta_H(\partial W)>0$, its diffraction spectrum may contain a non-trivial continuous component. In that case lack of maximal density may be interpreted as introducing some randomness into the system. Note however that the maximal density condition is not necessary for pure point diffraction. For example one may take a window with empty interior and consider a shift of the window which has empty intersection with the projected lattice. The existence of such a shift is seen by a Baire argument, compare e.g.~\cite{BMS} or \cite[Prop.~2.12]{HR15}. This will result in an empty weak model set, which is pure point diffractive.
\end{remark}

\subsection{Bragg peaks in Meyer sets}

Next, let $\Lambda\subset G$ be a Meyer set. Since any autocorrelation $\gamma$ of $\Lambda$ is positive definite, it is weakly almost periodic \cite[Sec.~11]{MoSt} and admits an Eberlein decomposition \cite[Eqn.~(8.28)]{ARMA}. Let $\gamma_S$ denote the strongly almost periodic part of $\gamma$. Then $(\widehat\gamma)_{pp}=\widehat{\gamma_S}$, see \cite[Sec.~10]{MoSt}. Since $\Lambda$ is a Meyer set, there exists a cut-and-project scheme $(G, H, \mathcal{L})$ and a positive and positive definite $h \in C_c(H)$ such that $\gamma_S=\omega_h$, see \cite[Prop.~12.1]{NS11}. Since $h$ is positive and positive definite,  $\widehat{h}$ is a finite measure and thus $\widehat{h} \in L^1(\widehat{H})$. Hence $h \in PK(H)$, and thus  Theorem~\ref{equiv PSF -diff-dens} yields an alternative proof of the following result.

\begin{theorem}\cite[Thm.~12.2]{NS11} \label{diff Meyer sets} Let $\Lambda\subset G$ be a Meyer set with autocorrelation $\gamma$, with $\sigma$-compact $G$. Then there exists a cut-and-project scheme  $(G , H, \cL)$ such that
\begin{displaymath}
\gamma_S=\omega_h, \qquad
(\widehat{\gamma})_{pp}= \omega_{\widecheck{h}}
\end{displaymath}
for some $h\in PK(H)$.
\qed
\end{theorem}
\begin{remark}
Hence the formula for the pure point part of the diffraction of an arbitrary Meyer set $\Lambda$ is a consequence of the PSF for some cut-and-project scheme in which $\Lambda$ is a subset of a model set. We note that the above arguments can also be applied to the pure point part of the diffraction of an arbitrary weighted Dirac comb with Meyer set support, which reproves \cite[Prop.~12.1]{NS11} using the PSF.
\end{remark}

\section*{Acknowledgment}

CR would like to cordially thank the Department of Mathematics and Statistics of MacEwan University for hospitality during stays in 2013 and 2015. We also thank Volker Strehl for inspiring discussions on the PSF and Christoph Schumacher for a very careful reading of the manuscript. Part of this work was supported by a RASCAF grant from MacEwan University, and part of the work was supported by NSERC with the grant number 03762-2014, and the authors are grateful for the support.

\end{document}